\documentclass[pra,preprint,superscriptaddress,aps,longbibliography]{revtex4-2}
\usepackage{natbib}
\usepackage{graphicx}
\usepackage{enumitem}
\usepackage{array}
\usepackage{amsmath, bm}
\usepackage{amsthm}
\usepackage{amssymb}
\usepackage{latexsym}
\usepackage{amsfonts}
\usepackage{mathrsfs}
\usepackage{color}
\usepackage{float}
\usepackage{url}
\usepackage{physics}
\usepackage{subcaption}
\usepackage{algorithm}
\usepackage{algpseudocode}
\newtheorem{theorem}{Theorem}

\newtheorem{lemma}{Lemma}

\usepackage{tikz}
\usepackage{quantikz}

\DeclareGraphicsExtensions{.pdf,.png,.jpg}
\usepackage[colorlinks=true, urlcolor=blue, citecolor=blue, 
linkcolor=blue, pdfusetitle=false]{hyperref}

\allowdisplaybreaks
\newcommand{\vmu}{\boldsymbol{\mu}}
\newcommand{\vpsi}{\boldsymbol{\psi}}

\newcommand{\vW}{\boldsymbol{W}}

\newcommand{\vvh}{\boldsymbol{h}}
\newcommand{\mzero}{\boldsymbol{0}_m}

\begin{document}
	\title{An iterative Bayesian variational quantum eigensolver with the von Mises--Fisher distribution}
	
	\author{Trung Minh Huynh}
	\affiliation{Department of Physics, Hanyang University, Seoul, 04763, Republic of Korea}
	\author{Gwangil An}
	\affiliation{Department of Physics, Hanyang University, Seoul, 04763, Republic of Korea}
	\author{Minsu Kim}
	\affiliation{Department of Physics, Hanyang University, Seoul, 04763, Republic of Korea}
	\author{Yu-Seong Jeon}
	\affiliation{Department of Physics, Hanyang University, Seoul, 04763, Republic of Korea}
	\author{Jinhyoung Lee}
	\email{hyoung@hanyang.ac.kr}
	\affiliation{Department of Physics, Hanyang University, Seoul, 04763, Republic of Korea}
	\affiliation{Center for Quantum Simulation, Korea Institute of Science and Technology (KIST), Seoul, 02792, Republic of Korea}
	
\begin{abstract}
	We develop an iterative Bayesian variational quantum eigensolver in which
	the unknown ground state, parametrised as a real unit vector, is inferred
	through a von Mises--Fisher distribution sequentially updated from
	Hadamard-test measurements on a single ancilla qubit. Building on a
	companion paper that treated a single success outcome via exact moment
	matching, we generalise the scheme to condition on $m$ success outcomes for the likelihood function
	per iteration. Since exact moment matching becomes impractical for $m>1$,
	we exploit the large-concentration regime, in which a first-order
	expansion of the log-likelihood yields a closed-form von Mises--Fisher
	posterior with $m$ entering only through a multiplicative weight. We
	prove that the two convergence properties of the single-outcome
	scheme---non-decreasing overlap and non-decreasing concentration---are
	preserved provided $m$ lies within an explicitly characterised range. Based on the theory, we propose a  measurement-based algorithm and illustrate the convergence on randomly generated three- and	four-qubit Hamiltonian matrices, on $\mathrm{HeH}^+$, and on $\mathrm{H}_2$.
\end{abstract}
	\maketitle	
	
\section{Introduction}
\label{sec:intro}

Determining the ground state of a quantum many-body Hamiltonian is a central
task in quantum chemistry as well as condensed-matter physics, and the exponential
growth of the Hilbert space renders exact classical diagonalisation infeasible
beyond small systems~\cite{Feynman.1982, Lloyd.1996, Cao.2019, Bauer.2020, McArdle.2020}. The variational quantum eigensolver (VQE) of Peruzzo
\emph{et al.}~\cite{Peruzzo.2014} became a leading near-term approach, yet its
deployment faces barriers that are structural rather than incidental: the attainable accuracy is strongly affected by the choice of ansatz and the classical optimiser ~\cite{TILLY20221, Tang.2021, Fedorov.2022, Bonet.2023, Lavrijsen.2020} ; convergence is
non-monotonic, with no step-by-step guarantee of improvement ~\cite{TILLY20221,Harwood.2022, You.2022, legnini.2026}; the number of
measurement settings grows with system size~\cite{TILLY20221, Verteletskyi.2020}; and, most fundamentally, the
gradient of the energy cost function can vanish exponentially with the number of
qubits---the \emph{barren-plateau}
phenomenon---an obstruction that persists even for gradient-free
optimisers~\cite{McClean2018,Cerezo.2021,Larocca.2025,Cunningham.2025, Arrasmith.2021}. These barriers share a common origin: optimization occurs on an ansatz-defined landscape, whose expressibility, depth, initialization, and symmetry properties can dominate trainability more than the target problem alone~\cite{Liu.2024, Wang.2024, Grimsley.2023, Leone.2024}.

A distinct route avoids the parameter landscape altogether by coupling the
target system to a single ancilla qubit and extracting spectral information from
a two-outcome measurement whose circuit complexity does not grow with the
Hilbert-space dimension. This underlies the eigenstate and spectrum algorithms~\cite{Wang.2016,Wang.2010}, the eigenstate-witness scheme of Santagati
\emph{et al.}~\cite{Santagati.2018}, and quantum phase
estimation~\cite{Dobsicek.2007, Wiebe.2016, Marco.2018, OBrien.2019, Lin.2022, Jakob.2026} or Hamiltonian parameters learning~\cite{Granade.2012, Jianwei.2017,Mirani.2024}. These methods, however, either require a problem-specific ansatz that reintroduces the very
design burden of VQE, or infer only a low-dimensional object such as an
eigenphase rather than the full state.

The companion manuscript~\cite{Companion.paper} laid the theoretical
groundwork for such a scheme by placing a full Bayesian posterior \emph{directly
	over the state vector}, represented as a von Mises--Fisher (vMF) distribution on
$S^{p-1}$ and refined iteratively from single-ancilla Hadamard-test
measurements. It thereby demonstrated the \emph{potential} to locate the ground
state through Hadamard tests alone---without the proliferation of measurement
settings, the ansatz design, or the optimiser.
Concretely, Bayesian inference with a vMF belief progressively sharpens our
knowledge of where the solution may lie: conditioning the prior on a
\emph{single} success outcome yields a closed-form update that provably
increases both the overlap with the ground state and the concentration of the
posterior at every iteration. This establishes a viable starting point rather
than a complete method.

In that scheme each iteration conditions on the minimal possible evidence---a
single outcome---so it is natural to ask whether supplying \emph{more}
information per iteration would accelerate convergence toward the ground state.
We make this concrete by conditioning each update on $m$ success outcomes, and
ask the two questions that form the core of this paper: does the algorithm still
locate the same solution when $m>1$, and, if so, can we identify an optimal $m$
for a practical implementation? Answering these questions requires a closed-form general-$m$ update, which the
moment-matching route of the companion paper cannot supply. We instead
exploit the \emph{large-concentration regime}, in which a first-order expansion
of the log-likelihood reduces the posterior to a closed-form vMF distribution. Based on the theory, we propose a practical algorithm and test with  $\mathrm{HeH}^+$, Hydrogen molecule, and randomly generated Hamiltonian matrices to investigate the theoretical predictions.

\section{Method}
\label{sec:method}

\subsection{Recap of the single-outcome framework}
\label{sec:recap}

We briefly recall the ingredients established in the companion
manuscript, referring the reader there for full details. We aim to find the ground state $|h_0\rangle$ and ground energy $\varepsilon_0$ of a non-degenerate Hamiltonian $H = \sum_{n=0}^{d-1} \varepsilon_n |h_n\rangle \langle h_n|$.  We assume that energy lower and upper bounds may be available. However, in this work we do not require any prior knowledge about eigen-states of the Hamiltonian.  The outcome probability function $p(x\mid|\psi\rangle) = \frac{1}{2}[1+(-1)^x\sum_{n=0}^{d-1}\cos\varepsilon_n t |\langle \psi | h_n\rangle|^2] =\frac{1}{2}[1+(-1)^x\langle \psi|V|\psi\rangle]$ where $V =\sum_{n=0}^{d-1}\cos\varepsilon_nt |h_n\rangle \langle h_n|$.
A pure state $|\psi\rangle=\sum_{k=1}^d(\psi_k^r+i\psi_k^i)|k\rangle$ in this work is identified with a real unit vector $\vpsi=(\psi_1^r, \psi_1^i, \psi_2^r, \psi_2^i,\cdots,\psi_n^r, \psi_i^d)\in S^{2d-1}$ ~\cite{Companion.paper}, and our
belief about it is a vMF distribution
$p_0(\vpsi\mid\vmu_0,\kappa_0)=C_{2d}(\kappa_0)\exp(\kappa_0\vmu_0^{\rm T}\vpsi)$,
with mean direction $\vmu_0$, concentration $\kappa_0$, and normalisation
$C_{2d}(\kappa_0)$ (see Ref.~\cite{Mardia.1999} for more details about this distribution). A Hadamard test on a single ancilla returns a binary outcome
$x\in\{0,1\}$ with likelihood
\begin{equation}
	p(x\mid\vpsi)=\tfrac12\bigl[1+(-1)^x\,\vpsi^{\rm T}\vW\vpsi\bigr],
	\label{eq:likelihood}
\end{equation}
where $x=0$ is the \emph{success} outcome and the real symmetric matrix $\vW$
stores the spectral information of the Hamiltonian with the same role as the matrix $V$~\cite{Companion.paper}. As in Ref.~\cite{Companion.paper}, we assume that with appropriate value of the time $t$, the energy levels satisfy
$\cos(\varepsilon_0 t)>\cos(\varepsilon_1 t)\cdots \cos(\varepsilon_{d-1}t)$,
the success probability $p(0\mid\vpsi)=\tfrac12(1+\vpsi^{\rm T}\vW\vpsi)$ is
maximised precisely when $\vpsi$ coincides with the ground state $\vmu_t$, the
leading eigenvector of $\vW$ with eigenvalue $\lambda_0=\cos(\varepsilon_0 t)$.

\begin{figure}[htpb]
	\centering
		\includegraphics[width=0.7\linewidth]{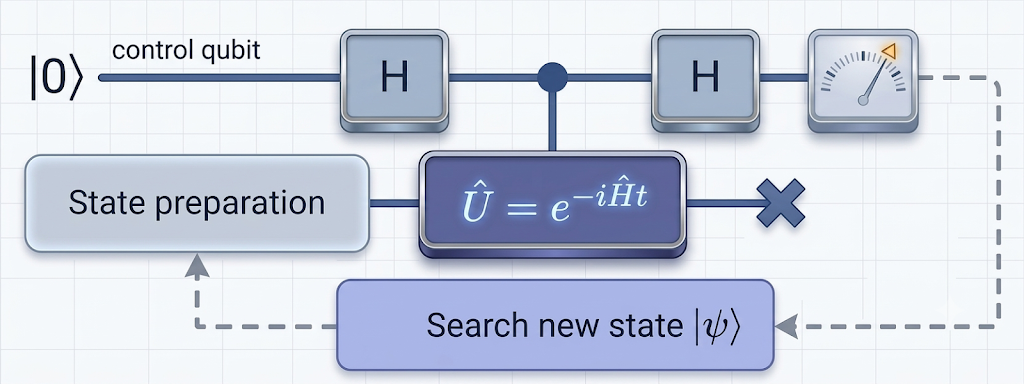}
	\caption{A typical Hadamard test measurement scheme with a control-unitary operator. We just perform a two-outcome measurement on the control qubit only.}
	\label{fig:meascheme}

\end{figure}

For a single success outcome, Bayes' theorem
$p_1(\vpsi\mid D)=z^{-1}p(D\mid\vpsi)p_0(\vpsi)$ was treated in
Ref.~\cite{Companion.paper} by moment matching: the exact posterior, a vMF
density multiplied by the quadratic factor $(1+\vpsi^{\rm T}\vW\vpsi)$, is
projected back onto the vMF family by matching its first two moments, yielding
coefficients $\alpha_0^{\rm exact}$ and $\beta_0^{\rm exact}$ expressed through
ratios of modified Bessel functions. The resulting update reads
\begin{equation}
	\vmu_1=\frac{\alpha_0^{\rm exact}\vmu_0+\beta_0^{\rm exact}\vW\vmu_0}
	{\lVert\alpha_0^{\rm exact}\vmu_0+\beta_0^{\rm exact}\vW\vmu_0\rVert},
	\label{eq:single_update}
\end{equation}
and two properties were proved: the overlap with the true solution does not decrease, and the concentration does not
decrease. Iterating therefore drives the estimate
monotonically toward the ground state.

\subsection{Relation to the present generalisation and choice of approximation}
\label{sec:relation}

The natural question is why the same moment-matching approach is not used for
$m>1$ outcomes. With the $m$-success likelihood
$L(\mzero\mid\vpsi)=p(0\mid\vpsi)^m$, the exact posterior is a vMF density
multiplied by $(1+\vpsi^{\rm T}\vW\vpsi)^m$. Moment matching would require
evaluating the moments of this posterior up to second order, which by the
recursive tensor-moment formula involves moments of the vMF distribution up to
order $2m$. The closed-form moments of the vMF distribution are known up to fourth-order ~\cite{Hillen.2017, Moakher.2015} but grow
rapidly in algebraic complexity with the order; consequently
the resulting expressions do not yield a compact closed form, making moment
matching impractical as a general-$m$ update rule.

The approach adopted here instead exploits the large-concentration regime. When
$\kappa_0$ is large, the prior is sharply peaked about $\vmu_0$ and the states
carrying appreciable probability all lie within a small neighbourhood of it.
Over such a neighbourhood $\ln p(0\mid\vpsi)$ varies weakly and a first-order
Taylor expansion is justified, reducing the $m$-outcome posterior to a vMF
distribution in closed form with coefficients that depend on $m$ only through a
simple multiplicative weight. This is a different approximation from the
moment-matching projection of ~\cite{Companion.paper}: the two approaches
use different algebras to map the non-vMF exact posterior back onto the vMF
family, and their coefficients are not the same in general. Both share the same
structural update $\vmu_1\propto\alpha_0\vmu_0+\beta_0\vW\vmu_0$, which tilts
the mean toward the dominant eigenvector of $\vW$. The large-concentration assumption is not a restriction on the
algorithm in practice: the concentration $\kappa$ is a free parameter of the vMF sampler
that the practitioner sets explicitly. One may choose $\kappa_0$ large enough
for the approximation to be accurate from the very first iteration.

\subsection{Generalisation to $m$ success outcomes}
\label{sec:generalisation}

Treating the $m$ measurements as independent, the $m$-success likelihood is
\begin{equation}
	L(\mzero\mid\vpsi)=p(0\mid\vpsi)^m
	=\Bigl[\tfrac12\bigl(1+\vpsi^{\rm T}\vW\vpsi\bigr)\Bigr]^{m}.
	\label{eq:mlikelihood}
\end{equation}
Substituting into Bayes' theorem and taking the logarithm,
\begin{equation}
	\ln p_1(\vpsi\mid 0)=\ln\frac{C_p(\kappa_0)}{z}
	+m\ln p(0\mid\vpsi)+\kappa_0\vmu_0^{\rm T}\vpsi,
	\label{eq:mlogpost}
\end{equation}
in which $m$ appears as a multiplicative weight on the nonlinear log-likelihood.
In the large-concentration regime the prior is sharply peaked about $\vmu_0$,
and expanding $\ln p(0\mid\vpsi)$ to first order about $\vmu_0$ gives the tangential gradient
\begin{equation}
	\bigl.\nabla_{\vpsi}\ln p(0\mid\vpsi)\bigr|_{\vpsi=\vmu_0}
	=\frac{1}{p(0\mid\vmu_0)}
	\bigl[\vW\vmu_0-(\vmu_0^{\rm T}\vW\vmu_0)\vmu_0\bigr].
	\label{eq:gradient}
\end{equation}
Inserting this into Eq.~(\ref{eq:mlogpost}) makes every term in the exponent
linear in $\vpsi$, so the posterior recovers the vMF form with natural parameter
\begin{equation}
	\kappa_1\vmu_1
	=\Bigl[\kappa_0-\frac{m\,\vmu_0^{\rm T}\vW\vmu_0}{p(0\mid\vmu_0)}\Bigr]\vmu_0
	+\frac{m}{p(0\mid\vmu_0)}\,\vW\vmu_0.
	\label{eq:natural_param}
\end{equation}
Writing the update compactly as
\begin{equation}
	\vmu_1=\frac{\alpha_0\vmu_0+\beta_0\vW\vmu_0}
	{\lVert\alpha_0\vmu_0+\beta_0\vW\vmu_0\rVert},
	\qquad
	\kappa_1=\lVert\alpha_0\vmu_0+\beta_0\vW\vmu_0\rVert,
	\label{eq:mupdate}
\end{equation}
the coefficients are
\begin{equation}
	\alpha_0=\kappa_0-\frac{m\,\vmu_0^{\rm T}\vW\vmu_0}{p(0\mid\vmu_0)},
	\qquad
	\beta_0=\frac{m}{p(0\mid\vmu_0)}.
	\label{eq:alphabeta}
\end{equation}

Now we obtain the two following theorems (see details in Appendix~\ref{app:convergence}).

\begin{theorem}[Non-decreasing overlap]
	\label{thm:overlap_m}
	Let $\vmu_1$ be obtained from a vMF prior $(\vmu_0,\kappa_0)$ through the
	$m$-success update of Eqs.~(\ref{eq:mupdate})--(\ref{eq:alphabeta}) with $m \leq \kappa_0/2$, and
$\vmu_t$ denotes the ground state.
	Then
	\begin{equation}
		R=\left(\frac{\vmu_t^{\rm T}\vmu_1}{\vmu_t^{\rm T}\vmu_0}\right)^2\geq 1,
		\label{eq:moverlap}
	\end{equation}
	with equality only when $\vmu_0 \in \operatorname{span}\{\vmu_t,\vmu_t^s\}$, where $\vmu_t^s$ denotes the conjugate representation, as defined in Ref.~\cite{Companion.paper}.
\end{theorem}

\begin{theorem}[Non-decreasing concentration]
	\label{thm:kappa_m}
	Under the same update, $\kappa_1\ge\kappa_0$, with equality only when $\vmu_0$
	is already an eigenvector of $\vW$.
	\label{eq:mkappa}
\end{theorem}

Theorems ~\ref{thm:overlap_m} and ~\ref{thm:kappa_m} are complementary and together characterise the full
Bayesian posterior at every iteration. Theorem 1 guarantees that the
mean vector $\vmu_n$ — the current best estimate of the ground state —
moves monotonically closer to the true solution $\vmu_t$. Theorem 2 guarantees that the concentration
parameter $\kappa_n$ increases monotonically: the posterior distribution
sharpens about the improving estimate, so the uncertainty shrinks
naturally as the evidence accumulates. Taken together, the two theorems describe a self-consistent Bayesian inference loop.

\subsection{Practical algorithm}
\label{sec:algorithm}

The closed-form update of Eqs.~(\ref{eq:mupdate})--(\ref{eq:alphabeta})
references $\vmu_n^{\rm T}\vW\vmu_n$ and $\vW\vmu_n$ explicitly. In any
realistic setting $\vW$ encodes an \emph{unknown} Hamiltonian and is
unavailable, so the update must be reconstructed from measurements alone.  The posterior mean is by definition the expectation of $\vpsi$ under the
posterior. Evaluating this expectation by Monte Carlo with states $\vpsi_k$
drawn from the current vMF prior (Appendix~\ref{app:montecarlo}) gives
\begin{equation}
	\mathbb{E}_{n+1}[\vpsi]
	\approx\frac{\sum_{k=1}^N p(x=0\mid\vpsi_k)^m\,\vpsi_k}
	{\sum_{k=1}^N p(x=0\mid\vpsi_k)^m}
	=\sum_{k=1}^{N}\alpha_k\,\vpsi_k,
	\quad
	\alpha_k=\frac{p(x=0\mid\vpsi_k)^m}
	{\sum_{j=1}^{N}p(x=0\mid\vpsi_j)^m}.
	\label{eq:practical_mean}
\end{equation}
The matrix $\vW$ disappears entirely, surviving only implicitly through the
measured success probabilities. From the estimated mean the updated
vMF parameters follow immediately,
\begin{equation}
	R_{n+1}=\bigl\lVert\mathbb{E}_{n+1}[\vpsi]\bigr\rVert,
	\qquad
	\vmu_{n+1}=\frac{\mathbb{E}_{n+1}[\vpsi]}{R_{n+1}},
	\label{eq:R_mu}
\end{equation}
with $\kappa_{n+1}$ recovered from $R_{n+1}$.

In the finite-sample regime, the choice of $m$ is critical. With only $N$
sampled states and noisy frequencies $f_k\approx p(x=0\mid\vpsi_k)$, a large
$m$ causes the importance weights $\alpha_k\propto f_k^{\,m}$ to collapse
onto the single state with highest measured frequency, stalling the sampler.
Conversely, too small $m$ yields uniform weights and no directional update.
We therefore select $m$ at each iteration by controlling the contribution
of the best candidate in the current sample set.

Let the measured frequencies be sorted in decreasing order,
$f_1\ge f_2\ge\cdots\ge f_N$, so that $f_1$ is the largest. The weight
carried by this leading candidate in Eq.~(\ref{eq:practical_mean}) is
\begin{equation}
	w_1(m)=\frac{f_1^{\,m}}{\sum_{k=1}^{N}f_k^{\,m}}
	=\Bigl[1+\sum_{k=2}^{N}\Bigl(\frac{f_k}{f_1}\Bigr)^{m}\Bigr]^{-1}.
	\label{eq:w1}
\end{equation}
Since $f_k/f_1\le1$, $w_1(m)$ increases monotonically from $1/N$ at $m=0$
toward $1$ as $m\to\infty$. We choose the positive integer $m$ for which
the leading candidate contributes a target fraction $\rho$ of the new mean
vector at $n$-th iteration,
\begin{equation}
	m^*=\min (\underset{m\in\mathbb{Z}_{\ge1}}{\arg\min}\;
	\bigl|\,w_1(m)-\rho\,\bigr|, \lfloor\kappa_n/2\rfloor),
	\qquad \rho\approx0.1,
	\label{eq:m_rho}
\end{equation}
i.e.\ $f_1^{\,m}/\sum_k f_k^{\,m}\approx0.1$. Because
$w_1(m)$ is monotone in $m$, Eq.~(\ref{eq:m_rho}) is solved by a simple
integer scan and costs no additional measurements. The value of $m^*$ is
recomputed at every iteration from the current frequencies, adapting
automatically to the spread of the sample set as the belief contracts. We
found $\rho\approx0.1$ to be stable across the systems tested; a systematic
scan over $\rho$ is left to future work.

Each iteration proceeds in three stages. \emph{First}, we \textbf{sample}: $N$
candidate states $\{\vpsi_k\}_{k=1}^{N}$ are drawn from
$p_n(\vpsi\mid\vmu_n,\kappa_n)$ (use the sampling method given
by~\cite{Kurz.2015}). \emph{Second}, we \textbf{measure}: for each
$\vpsi_k$ we perform $N_{\rm m}$ Hadamard-test measurements and record the
success frequency
\begin{equation}
	f_k=\frac{N_s}{N_{\rm m}}\;\approx\;p(x=0\mid\vpsi_k).
	\label{eq:frequency}
\end{equation}
\emph{Third}, we \textbf{update}: replacing probabilities in
Eq.~(\ref{eq:practical_mean}) by frequencies and selecting $m_n$ via
Eq.~(\ref{eq:m_rho}),
\begin{equation}
	\mathbb{E}_{n+1}[\vpsi]\approx\sum_{k=1}^{N}\alpha_k\,\vpsi_k,
	\qquad
	\alpha_k=\frac{f_k^{\,m_n}}{\sum_{j=1}^{N}f_j^{\,m_n}},
	\label{eq:update_freq}
\end{equation}
after which Eq.~(\ref{eq:R_mu}) gives $R_{n+1}$ and $\vmu_{n+1}$.

The concentration is updated in two regimes. In the \emph{global} regime it is
estimated from $R_{n+1}$ via the Sra approximation~\cite{Sra.2012}, a widely
used closed-form approximant to the maximum-likelihood estimate of the vMF
concentration that avoids inverting ratios of modified Bessel functions,
\begin{equation}
	\hat{\kappa}_{n+1}=\frac{R_{n+1}(2d-R_{n+1}^{2})}{1-R_{n+1}^{2}},
	\label{eq:sra}
\end{equation}
which drives $\kappa$ upward as the distribution contracts; the approximation is
accurate for moderate-to-large $R$ but becomes ill-conditioned as $R\to 1$. The
\emph{local} regime therefore caps the concentration at $\kappa_{\rm local}$,
\begin{equation}
	{\kappa}_{n+1}=\min\!\bigl(\hat{\kappa}_{n+1},\,\kappa_{\rm local}\bigr),
	\label{eq:kappa_local}
\end{equation}
guarding against ill-conditioning of the Sra estimate near $R\to 1$ and
preventing an overly sharp distribution from trapping the search at a local
optimum. The hyperparameters $N$, $N_{\rm m}$, and $\kappa_{\rm local}$ are not
optimised in this work; the reported resource counts should be regarded as
illustrative rather than as optimised requirements, and their systematic study
is deferred to future work.

\begin{algorithm}[H]
	\caption{Bayesian VQE with $m$ success outcomes (measurement-based)}
	\label{alg:bvqe}
	\begin{algorithmic}[1]
		\Require initial mean $\vmu_0$, concentration $\kappa_0$,
		samples $N$, measurements per state $N_{\rm m}$,
		ceiling $\kappa_{\rm local}$
		\Ensure ground-state estimate $\vmu_n$ and energy estimate
		\State $n \gets 0$
		\While{termination criterion not met}
		\State draw $N$ states $\{\vpsi_k\}$ from
		$p_n(\vpsi\mid\vmu_n,\kappa_n)$ \Comment{vMF sampling}
		\For{each sampled state $\vpsi_k$}
		\State measure $N_{\rm m}$ Hadamard tests; record successes $N_s$
		\State $f_k \gets N_s/N_{\rm m}$
		\Comment{$f_k \approx p(x=0\mid\vpsi_k)$}
		\EndFor
		\State compute $m$ from $\{f_k\}$
		\State $\alpha_k \gets f_k^{\,m}/\sum_{j}f_j^{\,m}$ for all $k$
		\State $\mathbb{E}_{n+1}[\vpsi] \gets \sum_{k}\alpha_k\vpsi_k$
		\State $R_{n+1} \gets \lVert\mathbb{E}_{n+1}[\vpsi]\rVert$,\quad
		$\vmu_{n+1}\gets\mathbb{E}_{n+1}[\vpsi]/R_{n+1}$
		\State $\kappa_{n+1}\gets
		\min(\hat{\kappa}_{n+1},\kappa_{\rm local})$
		\Comment{Sra approximation}
		\State $n\gets n+1$
		\EndWhile
		\State \Return $\vmu_n$ and energy from success frequency at $\vmu_n$
	\end{algorithmic}
\end{algorithm}
The iteration is terminated once a suitable stopping criterion is met (target
fidelity or energy tolerance, stabilisation of $\vmu_n$, or exhaustion of an
iteration budget), and the output is the final mean vector together with estimated energy  $\varepsilon_{est}= \cos^{-1}(2f_0-1)/t$, where $f_0$ is the relative frequency corresponding to the last mean vector. The whole procedure is summarised in Algorithm~\ref{alg:bvqe}.

Operationally, each iteration is a sample--reweight--recentre cycle: 
$N$ candidate directions are drawn from the current belief, reweighted 
by the $m$-success likelihood, and the belief is recentred on the 
resulting weighted mean. It can be considered as  a parametric variant of sequential 
importance sampling ~\cite{Liu.2001, Chopin.2002, Moral.2006} on $S^{p-1}$: the vMF prior acts as the proposal, 
the $m$-success likelihood supplies the importance weights, and the 
vMF refit at each step compresses the weighted information into a 
compact parametric belief. Because the belief state is kept compact 
as $(\boldsymbol{\mu}_n,\kappa_n)$, no long-lived particle cloud is 
required; the concentration $\kappa_n$ sharpens the proposal 
automatically as evidence accumulates, without requiring a hand-tuned 
annealing schedule.

\subsection{Comparisons with other works}
\label{sec:heh}

As a benchmark on a physically relevant system, we apply the method to the
Helium Hydride ion $\mathrm{HeH}^+$, the prototypical molecule of the original
VQE experiment~\cite{Peruzzo.2014}, and compare the number of trial states
required to reach the ground-state energy. In this case, the initial belief is uniform on $S^{2d-1}$ (meaning that $\kappa_0 = 0$). Since the uniform vMF has no mean direction, thus we fix an arbitrary reference vector $\vmu_0$ by normalising a random draw from $\mathcal{N}(\boldsymbol{0},\mathbf{I}_{2d})$, so that the first iteration is an unbiased random exploration. All reported results are averaged over such random initialisations.In their work, the ground-state energy converges after roughly fifty optimization steps. Our result for the same molecule, Fig.~\ref{fig:Peru_Santa}(a), plots the estimated energy $E_0$ against the cumulative number of trial states. It descends monotonically to the ground energy after about one hundred to one hundred and fifty trial states, then remains locked onto the exact value without oscillation. As the authors in ~\cite{Peruzzo.2014} employed Nelder-Mead optimizer, the number of trial states might be $100$ on average. This means that our method can find the solution with a comparable number of quantum states.

\begin{figure}[htpb]
	\centering
	\begin{subfigure}{0.52\textwidth}
		\includegraphics[width=\textwidth]{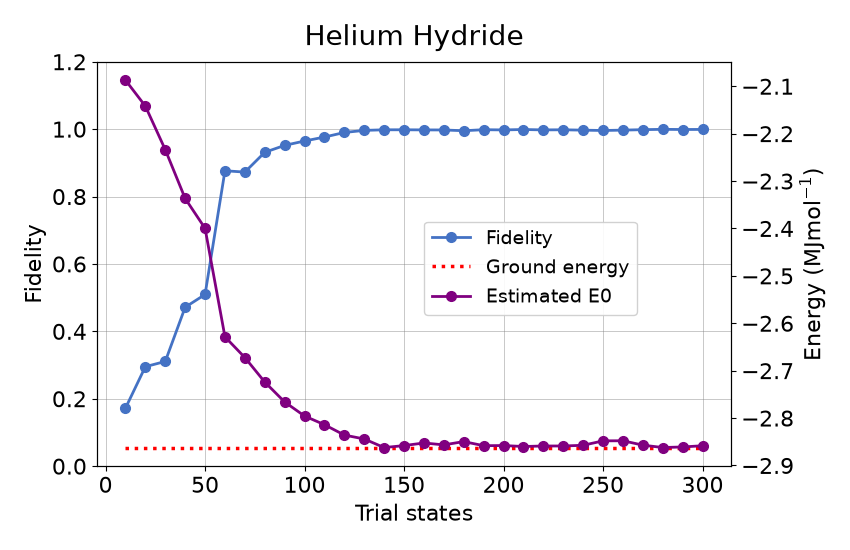}
	\end{subfigure}
		\hfill
	\begin{subfigure}{0.45\textwidth}
		\includegraphics[width=\textwidth]{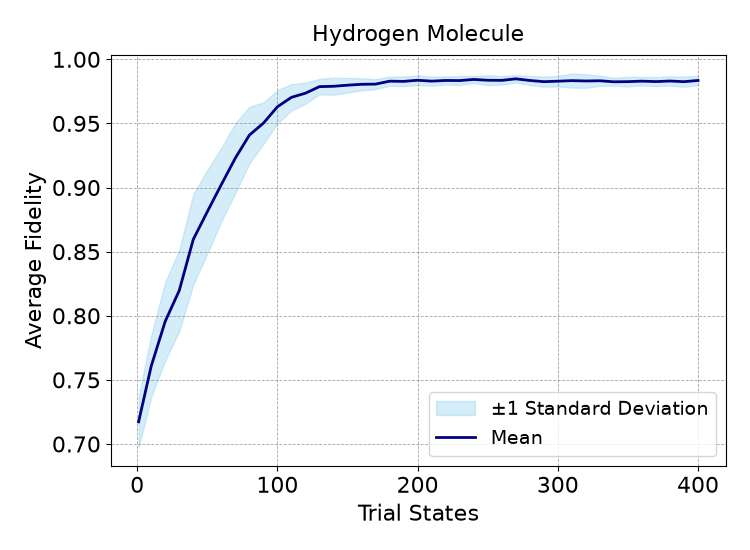}
	\end{subfigure}

	\caption{Searching ground state for Helium Hydride ion $\mathrm{HeH}^+$ and Hydrogen Molecule.
		(a) Estimated energy $E_0$ versus cumulative number
		of trial states in a single run, showing smooth monotonic descent to the true ground-state
		energy. In the first iteration, $N=10$ states are sampled uniformly from $S^{2d-1}$ because the initial concentration is set to $\kappa_0=0$. 
		(b) The starting guesses are initialized with fidelity $\approx 0.7$ to match the starting conditions of Ref.~\cite{Santagati.2018}, shaded band indicates statistical fluctuations. In both cases, we choose $\kappa_{\rm local}=500, N =10, N_m =100\,000$.}
			\label{fig:Peru_Santa}
\end{figure}

To place our method alongside the ancilla-based scheme, we
consider the four-qubit Hydrogen molecule $\mathrm{H}_2$, among the systems
studied by Santagati \emph{et al.}~\cite{Santagati.2018}. Because
Ref.~\cite{Santagati.2018} reports results for initial overlaps between $0.7$
and $0.8$, we initialise our sampler  to yield a
comparable starting fidelity of about $0.7$, so that the comparison reflects
algorithm efficiency rather than an initialisation advantage. The average
fidelity, shown in Fig.~\ref{fig:Peru_Santa}(b), climbs steeply past $0.98$ after
roughly one hundred to one hundred and fifty trial states, then remains
essentially flat, with a narrow band over random initial guesses.

\begin{table}[t]
	\centering
	\caption{Comparison of the present vMF-Bayesian scheme with the
		eigenstate-witness method of Santagati \emph{et al.}~\cite{Santagati.2018} on
		$\mathrm{H}_2$.}
	\label{tab:compare}
	\begin{tabular}{lll}
		\hline\hline
		Criterion & Santagati \emph{et al.}~\cite{Santagati.2018}
		& This work \\
		\hline
		Ansatz & problem-specific & general real-vector \\
		Convergence guarantee & none & monotone, step-by-step \\
		Measurement & phase estimation + witness & single Hadamard test \\
		Prior knowledge & required & Energy bounds \\
		Optimiser & swarm optimisation & none (Bayesian update) \\
		\hline\hline
	\end{tabular}
\end{table}

The comparison is not like-for-like, and the differences are best summarised
along the axes of Table~\ref{tab:compare}. The eigenstate-witness scheme of
Ref.~\cite{Santagati.2018} combines phase estimation with a swarm optimiser and a
problem-specific ansatz purpose-built for $\mathrm{H}_2$, encoding chemical and
symmetry information into a small number of parameters and restricting the search
to a low-dimensional manifold known to contain the solution; this requires prior
knowledge of the very state one seeks and must be redesigned for each system. Our
method employs the general real-vector parametrisation, searching over the entire
Hilbert space with no assumption about the target state, and provides a
step-by-step monotone-improvement guarantee that the swarm optimiser does not.
That we reach a fidelity above $0.98$ with a comparable number of trial states,
using far more parameters and no bespoke ansatz. The theoretical
convergence limit is unit fidelity; the gap between this limit and the observed
$\approx 0.98$ is attributable to finite $N$ and $N_{\rm m}$, the concentration
ceiling $\kappa_{\rm local}$, and residual approximation error from the
large-$\kappa$ linearisation at finite concentration. Closing this gap is a
matter of hyperparameter optimisation, which is left to future work.

\subsection{Randomly generated Hamiltonians}
\label{sec:random}

\begin{figure}[htpb]
	\centering
	\includegraphics[width=0.6\textwidth]{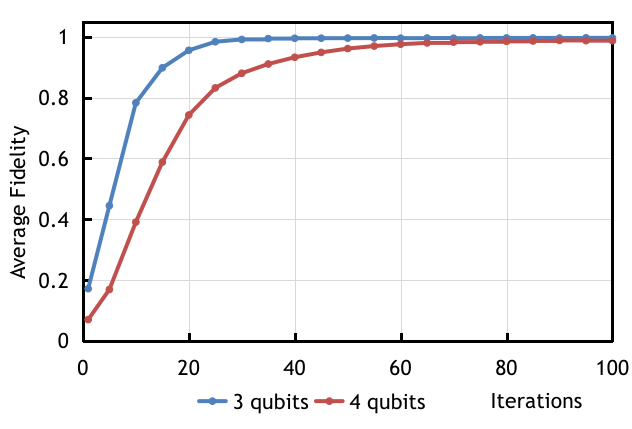}
	\caption{Average fidelity versus iteration number for randomly generated
		Hamiltonians with eigenvalues in $(0,\pi]$, averaged over one hundred
		independent instances. Blue: three-qubit systems ($N=300$,
		$N_{\rm m}=100{,}000$). Red: four-qubit systems ($N=500, N_m=100{,}000$). Both curves are
		monotonically increasing, consistent with the convergence guarantees.}
	\label{fig:random}
\end{figure}

To assess the method beyond specific molecular systems we test it on an ensemble
of randomly generated matrices whose eigenvalues are drawn uniformly from
$(0,\pi]$ with $t=1$. For each system size we generate one hundred independent
Hamiltonians, apply the algorithm from a random initial guess, and record the
average fidelity as a function of iteration number. The results are shown in
Fig.~\ref{fig:random} for three- and four-qubit systems.

The initial belief is uniform on $S^{2d-1}$, corresponding to $\kappa_0 = 0$. Since the uniform vMF distribution has no mean direction, the first iteration proceeds differently from the subsequent ones: we draw $N$ random vectors from a standard normal distribution $\mathcal{N}(\boldsymbol{0},\mathbf{I}_{2d})$, normalise each to unit length to obtain $N$ uniformly distributed states on the sphere, and feed them into the Hadamard tests. The resulting  mean, computed from the measured success frequencies via Eq.~(\ref{eq:update_freq}), directly yields the first updated mean direction $\vmu_1$ without any reference to a prior mean. From the second iteration onward, the algorithm proceeds as described above, with $\vmu_1$ and the corresponding $\kappa_1$ serving as the vMF parameters of the new belief. All reported results are averaged over such random initialisations.

It can be seen that in both cases the average fidelity is
monotonically increasing, with no non-monotonic step or plateau, consistent with
Theorems~\ref{thm:overlap_m} and~\ref{thm:kappa_m}. The gap between the
theoretical limit of unit fidelity and the observed values is due to the same
finite-sample and hyperparameter effects discussed. The
systematic dependence on system size is expected: as the Hilbert-space dimension
grows the search space expands and more iterations are needed. That the
algorithm behaves robustly across one hundred random instances provides evidence
that the monotone convergence is genuinely universal and does not rely on special
structure of the target state.

		\section{Conclusion}
		\label{sec:conclusion}
		
		We have generalised the Bayesian variational quantum eigensolver of the
		companion manuscript from a single success outcome ($m=1$) to a family of
		$m$-success updates, equivalent to conditioning the belief on those
		measurement records in which the favourable outcome is observed $m$ times.
		We proved that, in the large-concentration regime, the posterior obtained
		from a vMF prior and $m$ successes is again a vMF distribution admitting a
		closed-form update, in which $m$ controls the relative weight between the
		prior direction $\vmu_0$ and the Hamiltonian direction $\vW\vmu_0$. From
		this update we established two complementary convergence properties: the
		mean vector moves monotonically closer to the ground state, while the
		concentration parameter grows monotonically in average, so that the posterior sharpens
		and the uncertainty about the estimate contracts as evidence accumulates.
		In the ideal closed-form update these properties guarantee that the joint
		belief state $(\vmu_n,\kappa_n)$ converges monotonically toward the ground
		state. In the practical Monte Carlo realisation, finite sampling and the concentration
		ceiling $\kappa_{\rm local}$ limit the attainable precision, but the
		numerical experiments confirm that the monotone convergence trend persists.
	We further proposed a heuristic adaptive rule for $m^\star$, estimated on the fly from measured frequencies, which controls the contribution of the leading candidate to the importance-weighted mean while remaining within the admissible range that preserves both convergence guarantees.

		The scheme occupies a distinctive position among variational and Bayesian
		quantum algorithms. Because belief is placed directly on the state vector
		and refined by Bayesian conditioning, there is no parameterised circuit and
		no energy landscape to traverse, so the method can be less affected by the
		barren plateaus that obstruct gradient-based and gradient-free VQE as the
		system grows~\cite{McClean2018, Cerezo.2021, Arrasmith.2021}. On the molecular benchmarks it
		reaches a fidelity above $0.98$ for $\mathrm{H}_2$ with a trial-state count
		comparable to Ref.~\cite{Santagati.2018}, yet without the problem-specific
		ansatz that method requires, and it traces a smooth, monotone energy
		trajectory for $\mathrm{HeH}^+$ in contrast to the non-monotonic behaviour
		of the original VQE~\cite{Peruzzo.2014}.
%
%
%
		
		The method nonetheless has limitations. First, the general real-vector
		parametrisation searches the full state sphere $S^{2d-1}$, whose dimension
		grows exponentially with the number of qubits; the number of parameters, and
		hence the sampling cost, therefore rises rapidly with system size, and the
		demonstrations here are confined to three- and four-qubit systems. Second,
		a practical hardware realisation remains challenging: each trial state must
		be physically prepared and interrogated by many ancilla shots, so the total
		measurement budget is large and a controlled implementation on noisy devices
		is far from straightforward. Third, the algorithm carries several control
		hyperparameters --- the number of samples $N$, the shots per state
		$N_{\rm m}$, the concentration ceiling $\kappa_{\rm local}$, and the  target fraction governing $m^\star$ --- whose systematic optimisation, only preliminarily
		explored here, will be needed to close the gap between the observed fidelity
		and the theoretical limit of unity.

		Weighing these strengths against the limitations, the combination of
		Bayesian inference with the von Mises--Fisher distribution opens a promising
		research direction: it marries a powerful and well-understood statistical
		inference framework with a probability distribution naturally suited to the
		spherical geometry of quantum states. Because the update rule and the vMF
		belief are agnostic about the form of the Hamiltonian, the same
		sampling-and-update machinery may extend to Bayesian quantum state
		tomography, variational circuit learning, and quantum machine learning
		wherever a task reduces to optimising a success probability over a space of
		unit vectors. We expect this vMF-Bayesian approach to prove
		fruitful across a broad range of problems in quantum information and quantum
		computing.
		
		\begin{acknowledgments}
			JL was supported by the National Research Foundation of Korea (NRF) grant (No. 2022M3E4A1077369 and NRF-2023M3K5A1092818) funded by the Korea government (MSIT).
		\end{acknowledgments}
		
		\bibliography{refs2}

		\appendix
		
		\section{Proof of the convergence of the mean vector}
		\label{app:convergence}
		
		
		Let $\vW$ be the real symmetric matrix defined in Eq.~(\ref{eq:likelihood}),
		with spectral decomposition
		\begin{equation}
			\vW = \sum_{n=0}^{d-1}\lambda_n\bigl(\vvh_n\vvh_n^{\rm T}
			+ \vvh_n^s(\vvh_n^s)^{\rm T}\bigr),
			\qquad \lambda_n = \cos(\varepsilon_n t),
		\end{equation}
		where $\{\vvh_n, \vvh_n^s\}_{n=0}^{d-1}$ is an orthonormal eigenbasis of
		$\mathbb{R}^p$ ($p=2d$) and the eigenvalues are ordered as
		\begin{equation}
			\lambda_0 > \lambda_1 > \cdots > \lambda_{d-1}.
			\label{eq:eigen_order}
		\end{equation}
		Because the energies are scaled into $\varepsilon_n t \in (0,\pi]$, each
		eigenvalue lies in the interval
		\begin{equation}
			-1 \leq \lambda_n \le \lambda_0 < 1.
			\label{eq:eigen_range}
		\end{equation}
		The ground state $\vmu_t$ is any unit vector in the top eigenspace
		$\operatorname{span}\{\vvh_0,\vvh_0^s\}$, with $\vW\vmu_t = \lambda_0\vmu_t$.
		
		For an arbitrary unit vector $\vmu_0 \in S^{p-1}$ we define the spectral
		weights
		\begin{equation}
			w_n := (\vvh_n^{\rm T}\vmu_0)^2 + \bigl((\vvh_n^s)^{\rm T}\vmu_0\bigr)^2
			\;\ge\; 0,
			\qquad \sum_{n=0}^{d-1} w_n = 1,
			\label{eq:spectral_weights}
		\end{equation}
		so that for any integer $k \ge 1$,
		\begin{equation}
			\vmu_0^{\rm T}\vW^k\vmu_0 = \sum_{n=0}^{d-1}\lambda_n^k\,w_n.
			\label{eq:spectral_moments}
		\end{equation}
		
		Finally, recall the coefficients from Eq.~(\ref{eq:alphabeta}):
		\begin{equation}
			\alpha_0 = \kappa_0 - \frac{m\,\vmu_0^{\rm T}\vW\vmu_0}{p(0\mid\vmu_0)},
			\qquad
			\beta_0 = \frac{m}{p(0\mid\vmu_0)} > 0,
			\label{eq:alphabeta_app}
		\end{equation}
		where $p(0\mid\vmu_0) = \tfrac12(1 + \vmu_0^{\rm T}\vW\vmu_0) > 0$.

		The companion manuscript's proof of the non-decreasing overlap property
		exploited a structural inequality $\alpha_0 \ge \beta_0$ that follows from
		identities of modified Bessel functions under moment matching. Here the
		coefficients $\alpha_0,\beta_0$ have the different form
		\eqref{eq:alphabeta_app}, and we must impose an analogous condition.
		
		\begin{lemma}[A sufficient condition for $\alpha_0 \ge \beta_0$]
			\label{lem:alpha_ge_beta}
			For any $\vmu_0 \in S^{p-1}$ and any $m \ge 1$, if
			\begin{equation}
				m \;\le\; \frac{\kappa_0}{2},
				\label{eq:m_le_kappa}
			\end{equation}
			then $\alpha_0 \ge \beta_0 > 0$.
		\end{lemma}
		
		\begin{proof}
			Using \eqref{eq:alphabeta_app} and the identity
			$1 + \vmu_0^{\rm T}\vW\vmu_0 = 2\,p(0\mid\vmu_0)$,
			\begin{equation}
				\alpha_0 - \beta_0
				= \kappa_0 - \frac{m\bigl(\vmu_0^{\rm T}\vW\vmu_0 + 1\bigr)}
				{p(0\mid\vmu_0)}
				= \kappa_0 - \frac{m \cdot 2\,p(0\mid\vmu_0)}{p(0\mid\vmu_0)}
				= \kappa_0 - 2m.
			\end{equation}
			Thus $\alpha_0 \ge \beta_0$ iff $m \le \kappa_0/2$.
		\end{proof}

			The condition \eqref{eq:m_le_kappa} is \emph{sufficient} but not strictly
			necessary: for specific values of $\lambda_n$ the bracket
			$2\alpha_0 + \beta_0(\lambda_0 + \lambda_n)$ can remain non-negative even
			when $\alpha_0 < \beta_0$, provided $\lambda_0 + \lambda_n$ is not too
			negative. In the large-concentration regime $\kappa_0 \gg 1$ that the
			algorithm operates in (Sec.~\ref{sec:relation}), and with the moderate
			values of $m$ produced by the adaptive rule~\eqref{eq:m_rho}, condition
			\eqref{eq:m_le_kappa} is typically satisfied. We return to this point in
			the remark following Theorem~\ref{thm:overlap_m}.

		\subsection{Non-decreasing overlap (Theorem~\ref{thm:overlap_m})}
		
		We now prove Theorem~\ref{thm:overlap_m}. Let us recall the true solution vector $\vmu_t \in \operatorname{span}\{\vvh_0,\vvh_0^s\}$. We need to prove that
		\begin{equation}
			R := \left(\frac{\vmu_t^{\rm T}\vmu_1}{\vmu_t^{\rm T}\vmu_0}\right)^2 \geq 1.
		\end{equation}

		\begin{proof}[Proof of Theorem~\ref{thm:overlap_m}]
			Because $\vmu_t$ is an eigenvector of $\vW$ with eigenvalue $\lambda_0$,
			$\vmu_t^{\rm T}\vW = \lambda_0\vmu_t^{\rm T}$, and hence
			\begin{equation}
				\vmu_t^{\rm T}\bigl(\alpha_0\vmu_0 + \beta_0\vW\vmu_0\bigr)
				= (\alpha_0 + \beta_0\lambda_0)\,\vmu_t^{\rm T}\vmu_0.
				\label{eq:num_app}
			\end{equation}
			The inequality $R \geq 1$ is equivalent to
			$\Delta \ge 0$, where
			\begin{equation}
				\Delta := (\alpha_0 + \beta_0\lambda_0)^2
				- \lVert\alpha_0\vmu_0 + \beta_0\vW\vmu_0\rVert^2.
				\label{eq:Delta_def}
			\end{equation}
			Expanding the norm,
			\begin{equation}
				\Delta
				= 2\alpha_0\beta_0\bigl(\lambda_0 - \vmu_0^{\rm T}\vW\vmu_0\bigr)
				+ \beta_0^2\bigl(\lambda_0^2 - \vmu_0^{\rm T}\vW^2\vmu_0\bigr).
				\label{eq:Delta_expanded}
			\end{equation}
			Substituting the spectral decompositions
			\eqref{eq:spectral_moments} and factoring out
			$\beta_0\sum_n w_n(\lambda_0 - \lambda_n)$ gives
			\begin{equation}
				\Delta
				= \beta_0\sum_{n=0}^{d-1} w_n\,(\lambda_0 - \lambda_n)
				\bigl[2\alpha_0 + \beta_0(\lambda_0 + \lambda_n)\bigr].
				\label{eq:Delta_factored}
			\end{equation}
			Each factor is non-negative:
			
			\begin{enumerate}
				\item $w_n \ge 0$ by Eq.~\eqref{eq:spectral_weights};
				
				\item $\lambda_0 - \lambda_n \ge 0$ for every $n$, with strict
				inequality for $n \ge 1$, by the ordering \eqref{eq:eigen_order};
				
				\item for the bracket, since $\lambda_n > -1$ by
				\eqref{eq:eigen_range} we have $\lambda_0 + \lambda_n > \lambda_0 - 1
				\ge -2$, and therefore, using
				Lemma~\ref{lem:alpha_ge_beta} ($\alpha_0 \ge \beta_0 > 0$),
				\begin{equation}
					2\alpha_0 + \beta_0(\lambda_0 + \lambda_n)
					\;\ge\; 2\alpha_0 - 2\beta_0
					\;=\; 2(\alpha_0 - \beta_0)
					\;\ge\; 0.
					\label{eq:bracket_bound}
				\end{equation}
				This inequality is strict whenever $\lambda_0 + \lambda_n > -2$, i.e.,
				whenever $\lambda_0 > -1$ or $\lambda_n > -1$; since $\lambda_0 > -1$ by
				\eqref{eq:eigen_range}, it holds strictly for every $n$.
			\end{enumerate}
			
			Hence $\Delta \ge 0$, which gives $R_c \ge 1$.
			
			\emph{Equality case.} Since $\beta_0 > 0$ and every term in
			\eqref{eq:Delta_factored} has $w_n \ge 0$,
			$\lambda_0 - \lambda_n > 0$ for $n \ge 1$, and the bracket strictly
			positive by \eqref{eq:bracket_bound}, we have $\Delta = 0$ if and only if
			$w_n = 0$ for every $n \ge 1$. This is equivalent to $\vmu_0 \in
			\operatorname{span}\{\vvh_0,\vvh_0^s\}$. The
			update then acts as the scalar $\lambda_0$ on this two-dimensional
			subspace and does not rotate within it, so $\vmu_1 = \vmu_0$ up to the
			normalisation that is already in force.
		\end{proof}

			Lemma~\ref{lem:alpha_ge_beta} shows that $\alpha_0 \ge \beta_0$ holds
			whenever $m \le \kappa_0/2$. This condition is \emph{automatically
				satisfied} in the large-concentration regime that the algorithm operates
			in: $\kappa_0$ is chosen large by the practitioner
			(Sec.~\ref{sec:relation}), while the adaptive rule~\eqref{eq:m_rho} selects
			$m_n$ from the empirical frequencies and, in the regime of sharply peaked
			priors $\kappa_n \gg 1$, keeps $m_n$ moderate. If a particular run produces
			$m_n > \kappa_n/2$, one may simply cap $m_n$ at
			$\lfloor\kappa_n/2\rfloor$ without sacrificing the structural update; the
			resulting algorithm still satisfies Theorems~\ref{thm:overlap_m}
			and~\ref{thm:kappa_m}. The condition is therefore a minor bookkeeping
			constraint on the hyperparameter rather than an obstruction to the method.

		\subsection{Non-decreasing concentration (Theorem~\ref{thm:kappa_m})}

		Writing $b=\beta_0=m/p(0\mid\vmu_0)$ and squaring the norm of the natural
		parameter,
		\begin{equation}
		\kappa_1^2=\kappa_0^2+b^2
		\bigl[\vmu_0^{\rm T}\vW^2\vmu_0-(\vmu_0^{\rm T}\vW\vmu_0)^2\bigr].
		\label{eq:kappa_growth}
		\end{equation}
		The bracket is non-negative by Cauchy--Schwarz,
		$(\vmu_0^{\rm T}\vW\vmu_0)^2\le\vmu_0^{\rm T}\vW^2\vmu_0$,
		so $\kappa_1\ge\kappa_0$, with equality only when $\vW\vmu_0\parallel\vmu_0$.
		\hfill$\square$

		\section{Monte Carlo derivation of the practical update}
		\label{app:montecarlo}

		The mean vector of the posterior is by definition
		\begin{align}
			\mathbb{E}_{n+1}[\vpsi]
			&=\frac{1}{z}\int d\vpsi\,\vpsi\,L(\mzero\mid\vpsi)\,p_n(\vpsi),
			\qquad
			z=\int d\vpsi\,L(\mzero\mid\vpsi)\,p_n(\vpsi),
			\label{eq:posterior_mean_def}
		\end{align}
		where $p_n(\vpsi)$ is the current vMF prior. This expression references $\vW$
		only through the likelihood, which is the success probability of the Hadamard
		test. Applying the law of large numbers with $N$ states $\vpsi_k$ drawn from
		$p_n(\vpsi)$,
		\begin{equation}
			z_n\approx\frac{1}{N}\sum_{k=1}^{N}L(\mzero\mid\vpsi_k),
			\qquad
			\mathbb{E}_{n+1}[\vpsi]
			\approx\frac{\sum_{k=1}^{N}L(\mzero\mid\vpsi_k)\,\vpsi_k}
			{\sum_{k=1}^{N}L(\mzero\mid\vpsi_k)},
			\label{eq:sample_mean}
		\end{equation}
		where the $1/N$ cancels. Inserting
		$L(\mzero\mid\vpsi)=p(x=0\mid\vpsi)^m$ gives
		\begin{equation}
			\mathbb{E}_{n+1}[\vpsi]
			\approx\sum_{k=1}^{N}\alpha_k\,\vpsi_k,
			\qquad
			\alpha_k=\frac{p(x=0\mid\vpsi_k)^m}
			{\sum_{j=1}^{N}p(x=0\mid\vpsi_j)^m},
			\label{eq:practical_mean_app}
		\end{equation}
		which is Eq.~(\ref{eq:practical_mean}) of the main text, with $\sum_k\alpha_k=1$
		by construction. 
		
		\end{document}